\documentclass[a4paper,reqno,12pt]{article}
\usepackage{mathrsfs}

\usepackage{amsmath}
\usepackage{amsfonts}
\usepackage{amssymb}
\usepackage{amsthm}
\usepackage[mathscr]{eucal}
\usepackage[hmargin=2.5cm,vmargin=2.5cm]{geometry}

\usepackage{graphicx}
\usepackage{indentfirst}
\usepackage{newlfont}

\numberwithin{equation}{section}

\newcommand{\real}{\mathbb{R}}

\theoremstyle{plain}
\newtheorem{theorem}{Theorem}[section]
\newtheorem{lemma}[theorem]{Lemma}

\newtheorem{definition}[theorem]{Definition}
\newtheorem{proposition}[theorem]{Proposition}

\theoremstyle{remark}
\newtheorem{remark}[theorem]{Remark}

\begin{document}

\title{
Dynamics of 2D Stochastic non-Newtonian fluids
driven by fractional Brownian motion
\footnote{Support by NSFC (No. 10971225) and SRF for ROCS, SEM, China}}
\author{Jin Li  and Jianhua Huang\\
Department of Mathematics, National University of Defense Technology,\\
Changsha 410073,  P.R.China }
%$^a$ \textsl{\small Corresponding author E-mail: talent\_tim@yahoo.com.cn} }
%%\thanks{talent\_tim@yahoo.com.cn}

\date{\today}

\maketitle

\abstract{
A 2D Stochastic incompressible non-Newtonian fluids driven by
fractional Bronwnian motion with Hurst parameter $H \in (1/2,1)$ is studied.
The Wiener-type stochastic integrals are introduced for infinite-dimensional fractional Brownian motion.
Four groups of assumptions, including the requirement of Nuclear operator or Hilbert-Schmidt operator,
are discussed.
The existence and regularity of stochastic convolution for the corresponding additive linear 
stochastic equation are obtained under each group of assumptions.
Mild solution are then obtained for the non-Newtonian systems by the modified fix point theorem in
the selected intersection space.
When the domain is square, the random dynamical system generated by non-Newtonian systems 
has a random attractor under some condition on the spectrum distribution of 
the corresponding differential operator.
}



\section{Introduction}

In this paper, we consider the following 2D stochastic non-Newtonian fluids driven by
fractional Brownian motion (fBm, for short) with Hurst parameter $H \in (\frac{1}{2},1)$:
\begin{align}\label{Problem: initial}
     \frac{\partial{u}}{\partial{t}} + (u \cdot \nabla)u + \nabla p
    & = \nabla \cdot (\mu (u)e-2\mu_1 \Delta e)+ \Phi \frac{dB^H(t)}{dt}
     \quad & x \in \mathcal{O}, \ t>0  \\
    \nabla \cdot u &=0,
     \quad & x \in \mathcal{O}, \ t>0  \\
     u &= 0, \ \tau_{ijk} \eta_j \eta_k = 0,
     \quad & x \in \partial \mathcal{O}, \ t \geq 0  \\
     u &= u_0, \quad &x \in \mathcal{O}, t=0  \label{Problem: initial end}
\end{align}
where $\mathcal{O}$ is a smooth bounded domain of $\mathbb{R}^2$,
$\tau_{ij}$ is the components of the stress tensor,
$\tau_{ijk}$ is the components of the first multipolar stress tensor,
and $p$ is the pressure.
$e_{ij}$ are the components of the rate of deformation tensor, i.e.
\begin{equation}
    e_{ij} = \frac{1}{2} \left( \frac{\partial u_i}{\partial x_j}
        + \frac{\partial u_j}{\partial x_i} \right) .
\end{equation}
$\epsilon, \mu_0, \mu_1 >1$ and $\alpha, 0< \alpha \leq 1$, are constitutive parameters.
The noise, modeled by the formal derivative of fBm, enters linearly in the equation.
This model describes the simplest case
when a 2D isothermal, nonlinear, incompressible bipolar viscous fluids
perturbed by the noise with long range dependence.

There are many works concerning the unique existence, regularity and asymptotic behavior of solution
to the incompressible non-Newtonian fluids or its associated versions (see e.g.
\cite{Bellout Bloom Necas, Bloom Hao existence, Guo Guo, Ladyzhenskaya viscousflow, Zhao Duan, Zhao Zhou}).
For instance, Zhao and Duan \cite{Zhao Duan} established the existence of random attractor
for some non-Newtonian model with additive noise.
Noise are intrinsic effects in a variety of settings and spatial scales.
It could be mostly obviously influential at the microscopic and small scales
but indirectly it plays a vital role in microscopic phenomena.
We study fBm other than the standard commonly used Wiener process as the source of noise.
The fBm is a family of Gaussian processes that is indexed by the Hurst parameter $H \in (0,1)$.
These processes with values in $\mathbb{R}$ were introduced by Kolmogorov \cite{Kolmogorov}
and some useful properties of these process were given by Mandelbrot and Van Ness \cite{Mandelbrot VanNess}.
For $H \neq \frac{1}{2}$ the fBm is not a semi-martingale
and the increments of the process are not independent.
These properties can be used in modeling ``cluster'' phenomena (system with memory and persistence)
such as hydrology \cite{Hurst}, economic data \cite{Mandelbrot}
and telecommunications \cite{Leland Taqqu Willinger}.
Since there are limited publish works 
(\cite{Duncan Maslowski Duncan, Duncan Maslowski PasikDuncan, Garrido-Atienza Lu Schmalfuss, Maslowski Schmalfus, Tindel Tudor Viens})
on infinite-dimensional fBm-driven equations, 
we first introduce the Wiener-type stochastic integral with respect to one-dimensional fBm in section 2.
Then we discuss the existence and regularity of infinite-dimensional stochastic convolution
under several existed assumptions
(such as the requirement of Nuclear operator or Hilbert-Schmidt operator) for this model.
We also study a fundamental example by some subtle calculus on the spectrum of differential operator
without compactness assumption on other parameter.
Inspired by \cite{Lions Magenes} ,in section 3 we obtain the mild solution for \eqref{Problem: initial} by
a modified fix point theorem which needs careful estimation on the selected intersection space.
Since the fractional noise is not Markovian,
the solution to \eqref{Problem: initial} can not be expected to define a Markov process.
Therefore, the approach to study the long-time behavior of solution via invariant measure is not
an option here.
In particular, Corollary 4.4 in \cite{Crauel Flandoli} can not be applied here.
Consequently, our analysis in dynamics is instead based on the framework of random dynamical system(RDS),
which more or less requires the driving process to have stationary increments.
Since the foundational work in \cite{Crauel Flandoli, Crauel Debussche Flandoli}
the long time behavior of SPDE has been extensively investigated by means of proving the
existence of a random attractor 
(e.g. \cite{Bates Lu Wang, Gess Liu Rochner, Garrido-Atienza Maslowski Schmalfuss, Zhao Duan}).
In this work we use the stationary generalized stochastic integral
(known as fractional Ornstein-Uhlenback process) to construct RDS associated to \eqref{Problem: initial}
and obtain a random attractor under certain condition in section 4.

Compared with \cite{Zhao Duan}, the difficulties in our work are:
(i) Unlike the finite-dimensional white noise, handling infinite-dimensional fBm needs addition
assumptions on parameter (such as the requirement of Nuclear operator or Hilbert-Schmidt operator)
and subtle estimation on spectrum of differential operator.
(ii) Unlike the classical Ito integral with Brownian motion,
the stochastic integrals with fBm are more complicated.
There are several type of integration for fBm and each only preserve part of the properties of integrator.
(iii) The principal method for creating weak solutions to stochastic differential equations is
transformation of drift via the Girsanov theorem (cf. \cite{Karatzas Shreve}).
However, significant difficulties arise when the application of infinite-dimensional Girsanov theorem
is used for this nonlinear fBm-driven stochastic fluid equations.
Thus, we have to use the mild solution.
(iv) We emphasize that the cocycle property of RDS has to be satisfied for any $\omega \in \Omega$.
It is not sufficient if this equation is only true almost surely
where the exceptional set may depend on space variable or time varialbe.
Thus, we have to switch fBm to its incomplete equivalent canonical realization
and restrict the fBm parameter $H> \frac{1}{2}$ which allow us to solve the equation
using Wiener-type stochastic integral for deterministic integrands understood in a pathwise way.

We express this problem by the standard mathematical setting.
We use lowercase $c_i, i\in \mathbb{N}$ for global constants and
capital $C$ for local constants which may change value from line to line.
Denote

$\mathcal{V} = \left\{ \phi = (\phi_1, \phi_2) \in \left( C_0^{\infty} \left( \mathcal{O} \right) \right)^2 :
\ \nabla \cdot \phi =0 \text{ and } \phi=0 \text{ on } \partial \mathcal{O} \right\}$.

$H = $ the closure of $ \mathcal{V}$ in $\left( L^2\left(\mathcal{O}\right)\right)^2$ with norm $|\cdot|$.
%Evidently $|u|=|u|_0$ for all $u \in H$.

$V = $ the closure of $\mathcal{V}$ in  $\left( H^2(\mathcal{O})\right)^2$ with norm $|\cdot|_V$.

Thus, $(V,H,V')$ is a Gelfand triple.
Define a bilinear form $a(\cdot, \ \cdot):V \times V \rightarrow \mathbb{R}$,
\begin{equation}
    a(u,v)=\frac{1}{2}(\triangle u, \triangle v).
\end{equation}
According to \cite{Bloom Hao existence} Lemma 2.3,
we can use Lax-milgram Theorem to define
$A \in \mathcal{L}(V,V')$:
\begin{equation}
    <Au,v>=a(u,v) \quad \forall \ u,v \in V.
\end{equation}
And we have
  \begin{enumerate}\label{Proposition: A1}
    \item[(i)] Operator $A$ is an isometric form $V$ to $V'$.
    Furthermore, let $D(A)=\{ u \in V : a(u,v)=(f,v), f \in H \}$.
    Then $A \in \mathcal{L}(D(A),H)$ is an isometric form $D(A)$ to $H$.

    \item[(ii)] Operator $A$ is self-adjoint positive with compact inverse.
    By Hilbert Theorem, there exist eigenvectors $\{ e_i \}_{i=1}^{\infty} \subset D(A)$
    and eigenvalues $\{ \lambda_i \}_{i=1}^{\infty}$ s.t.
    \begin{align}
      & Ae_i=\lambda e_i, \quad e_i \in D(A), \quad i=1,2, \cdots \\
      & 0<\lambda_1 \leq \lambda_2 \leq \cdots \leq\lambda_i \leq \cdots,
      \quad \lim_{i\rightarrow \infty} \lambda_i = \infty.
    \end{align}
    And $\{ e_i \}_{i=1}^{\infty}$ form an orthonormal basis for $H$.
  \end{enumerate}
Define the trilinear form:
\begin{equation}
    b(u,v,w)=\sum_{i,j=1}^{2} \int_{\mathcal{O}} u_i \frac{\partial v_j}{\partial x_i} w_j dx
    \quad \forall u,v,w \in H_0^1 (\mathcal{O}).
\end{equation}
and the functional $B(u,v) \in V'$:
\begin{equation}
    <B(u,v),w>=b(u,v,w) \quad \forall w \in V
\end{equation}
Denote $B(u):=B(u,u)\in V'$.
Define $N(u) \in V'$ as
\begin{equation}
    <N(u),v>=\int_{\mathcal{O}} \mu(u) e_{ij}(u) e_{ij}(v) dx
    \quad \forall v \in V.
\end{equation}
For the properties of operator $A$, $B$ and $N$ we refer to \cite{Zhao Zhou}.
Comprehensively, we have the following abstract evolution equation from problem
\eqref{Problem: initial}-\eqref{Problem: initial end}:
\begin{equation}\label{Problem: differential form}
  \left\{
  \begin{split}
    du+ \left( 2\mu_1 Au +B(u)+N(u)  \right) dt &= \Phi dB^H(t),  \\
    u(0) & =u_0.
  \end{split}  \right.
\end{equation}
Without loss of generality, we set $\mu_1 = 1$ in the sequel.

\section{FBm and random dynamical systems}
Since the derivative of fBm exists almost nowhere,
we seek the solution in the integral form.
There are several approach to define an integral for one-dimensional fBm
and each has its advantage
(for a useful summary we refer to \cite{Biagini Hu Oksendal Zhang}).
In this paper, we adopt the Wiener integrals since they deal with the simplest case of
deterministic integrands.
However, the assumption on noise driven by infinite dimensional fBm varies.
Therefore, we first introduce the general framework of the Wiener-type stochastic integral with
respect to infinite dimensional fBm, then discuss three sets of assumptions and obtain the desired result.

Let $\beta^H(t)$ be the one-dimensional fBm with Hurst parameter $H$.
Throughout this paper we only consider the case $H \in (\frac{1}{2},1)$.
For a survey of Winer-type stochastic integral we refer to \cite{Biagini Hu Oksendal Zhang}.
By definition $\beta^H$ is a centered Gaussian process with covariance
\begin{equation}
    R(t,s)=\mathbb{E}(\beta^H(t) \beta^H(t))= \frac{1}{2} (t^{2H}+s^{2H}-|t-s|^{2H}).
\end{equation}

\begin{remark}
  Wether the symbol $H$ represents the Hurst parameter or Hilbert space depends on context.
\end{remark}

$\beta^H$ has the following Wiener integral representation:
\begin{equation}\label{Equation: Wiener inte repre}
    \beta^H(t) = \int_0^t K_H(t,s) d \beta (s),
\end{equation}
where $\beta$ is a Wiener process, and $K_H(t,s)$ is the kernel given by
\begin{equation}\label{Equation: defi of kernel}
    K_H(t,s)
    = c_H \left( \frac{t}{s} \right)^{H-\frac{1}{2}}  (t-s)^{H- \frac{1}{2}}
      + s^{\frac{1}{2}-H} F(\frac{t}{s}).
\end{equation}
$c_H$ is a constant and
\begin{equation}
    F(z) = c_H \left( \frac{1}{2} - H \right)
             \int_0^{z-1} r^{H - \frac{3}{2}} \left( 1-(1+r)^{H- \frac{1}{2}} \right) dr.
\end{equation}
By \eqref{Equation: defi of kernel} we obtain
\begin{equation}
    \frac{\partial K_H}{\partial t}(t,s)
      = c_H (H-\frac{1}{2}) (t-s)^{H-\frac{3}{2}} \left( \frac{s}{t}\right)^{\frac{1}{2}-H}.
\end{equation}
Denote by $\mathscr{E}$ the linear space of step function of the form
\begin{equation}
    \phi(t) = \sum_{i=1}^{n} a_i 1_{(t_i,t_{i+1} ]} (t)
\end{equation}
where $n \in \mathbb{N}$, $a_i \in \real$ and by $\mathscr{H}$ the closure of $\mathscr{E}$
with respect to the scalar product
\begin{equation}
    <1_{[0,t]}, 1_{[0,s]}>_{\mathscr{H}} = R(t,s)
\end{equation}
For $\phi \in \mathscr{E}$ we define its Weiner integral
with respect to the fBm as
\begin{equation}
    \int_0^T \phi(s) d \beta^H(s) = \sum_{i=1}^{n} a_i (\beta^H_{t_{i+1}} - \beta^H_{t_{i}}).
\end{equation}
The mapping
\begin{equation}
    \phi = \sum_{i=1}^{n} a_i 1_{(t_i,t_{i+1} ]}
      \rightarrow \int_0^T \phi(s) d \beta^H(s)
\end{equation}
is an isometry between $\mathscr{E}$ and the linear space $span \{ \beta^H(t), 0\leq t \leq T \} $
viewed as a subspace of $L^2(0,T)$
and it can be extended to an isometry between $\mathscr{H}$
and the $\overline{span}^{L^2} \{ \beta^H(t), 0\leq t \leq T  \}$.
The image on an element $\Psi \in \mathscr{H}$ by this isometry is called the Wiener integral of
$\Psi$ with respect to $\beta^H$.

Next we give a characterization of the so-called reproducing kernel Hilbert space $\mathcal{H}$.
Consider space $L^2(0,T)$ equipped with the twisted scalar product
\begin{equation}
    <f,g>_{\mathcal{H}} := H(2H-1) \int_0^T \int_0^T f(s)g(t) |s-t|^{2H-2} ds dt.
\end{equation}
Since
\begin{equation}
    <1_{[0,t]}, 1_{[0,s]}>_{\mathscr{H}}
    = H(2H-1) \int_0^t \int_0^s |u-r|^{2H-2} du dr
    = R(t,s)
\end{equation}
We have that $\mathcal{H}$
can be represented by the closed of $L^2(0,T)$ with respect to the twisted scalar product.
Namely,
$\mathcal{H}= \overline{(L^2(0,T),<,>_{\mathcal{H}})} = \overline{( \mathscr{E},<,>_{\mathcal{H}})}$.
In \cite{Pipiras Taqqu} it is shown that the elements of $\mathcal{H}$ may not be functions
but distributions of negative order.
By \cite{Memin Mishura Valkeila} we have the following inclusion
$L^2(0,T) \subset L^{1/H}(0,T) \subset  \mathcal{H}$.
We now introduce the linear operator $K^*_H$ defined on $\phi \in \mathscr{E}$ as follows:
\begin{equation}\label{Equation: K* representation}
    (K^*_H \phi )(s) =
      \int_s^T \phi(t) \frac{\partial K_H}{\partial t}(t,s) dt.
\end{equation}
We refer to \cite{Alos Mazet Nualart} for the proof of the fact that
$K^*_H$ is an isometry between the space $\mathscr{E}$ and $L^2(0,T)$ that
can be extended to the Hilbert space  $\mathscr{H}$ and $L^2(0,T)$, i.e.,
$\mathcal{H} = (K_H^*)^{-1} (L^2 (0,T))$.
As a consequence, we have the following relationship between the Wiener integral with respect
to fBm and the Wiener integral with respect to the Wiener process:
\begin{equation}\label{Equation: fBm represetation by Bm}
    \int_0^t \phi(s)d\beta^H(s) = \int_0^t (K^*_H \phi)(s) d \beta (s)
\end{equation}
for every  $\phi \in \mathscr{H}$ if and only if $K^*_H \phi \in L^2(0,T)$.
Since we work with Wiener integral over Hilbert space $H$,
we have that if $u \in L^2(0,T;H)$ is a deterministic function,
then the relation \eqref{Equation: fBm represetation by Bm} holds
and the Wiener integral on the righthand side being well defined in $L^2(\Omega;H)$
if $K^*_H u$ belongs to $L^2(0,T;H)$.

%infinite dimensional fBm and stochastic integration

Next we introduce the infinite dimensional fractional Brownian motion
and corresponding stochastic integration.
Let $Q$ be a self-adjoint and positive linear operator on $H$.
Assume that there exists a sequence of nonnegative numbers
$\{ \widetilde{\lambda}_i \}_{i \in \mathbb{N}}$ s.t.
\begin{equation}
  Q e_i =  \widetilde{\lambda}_i e_i,  \quad i=1,2,\cdots
\end{equation}
We formally define the infinite dimensional fBm on $H$ with covariance operator $Q$ as
  \begin{equation}
    B^H(t)=\sum^{\infty}_{i=1} \sqrt{\widetilde{\lambda}_i} e_i \beta_i^H(t)
  \end{equation}
  where $\{ \beta_i^H(t) \}_{i \in \mathbb{N}}$ is a sequence of real stochastically independent
  one-dimensional fBm's.
This process, if convergence, is a $H$-valued Gaussian process, it starts from 0,
has zaro mean and covariance
\begin{equation}
   \mathbb{E}(B^H(t) B^H(t))=  R(t,s)Q .
\end{equation}
Let $(\Phi_s)_{0\leq s \leq T}$ be a deterministic function with values in $\mathcal{L}(H)$,
the space of all the bounded linear operator from  $H$ to $H$.
The stochastic integral of $\Phi$ with respect to $B^H$ is formally define by
\begin{equation}\label{Equation: defi of W inte of inf fBm}
    \int_0^t \Phi_s dB^H(s)
     := \sum_{i=1}^{\infty}
       \sqrt{\widetilde{\lambda}_i}  \int_0^t \Phi_s e_i d \beta_i^H (s)
      =  \sum_{i=1}^{\infty}
       \sqrt{\widetilde{\lambda}_i}  \int_0^t (K^*_H(\Phi e_i))_s d \beta_i (s)
\end{equation}
where $\beta_i$ is the standard Brownian motion used to
represent $\beta^H$ as in \eqref{Equation: Wiener inte repre}.
The above sum may not converge.
However, as we are about to see,
the linear additive stochastic equation can have a mild solution
even if $ \int_0^t \Phi_s dB^H(s) $ is not properly defined as a $H$-valued Gaussian random variable.

Back to non-Newtonian systems.
Our goal is to find a mild solution of problem \eqref{Problem: differential form}.
Namely, seek a function  $u \in C([0,T];H)\cap L^2(0,T;V)$
such that the following integral equation holds for all $\omega \in \Omega$
  \begin{equation}\label{Equation: integral form}
    u(t)=S(t)u_0 -\int_0^t S(t-s)B(u(s))ds -\int_0^t S(t-s) N(u(s))ds
    + \int_0^t S(t-s)\Phi dB^H(s).
  \end{equation}
Here $S(t):=e^{-tA}=\int_{0}^{\infty} e^{-t\lambda} d E_{\lambda}$,
is an analytic semigroup generated by $A$
Since $A$ is a densely defined self-adjoint bounded-below operator in Hilbert space $H$
and hence a sector operator (see, for instance, \cite{DanielHenry} section 1.3).
The first and second integral are Bochner integral,
the last integral is the Wiener-type stochastic integral
defined by \eqref{Equation: defi of W inte of inf fBm}.
By the form of integral equation, we care the stochastic convolution $\int_0^t S(t-s)\Phi dB^H(s)$
other than $\int_0^t \Phi dB^H(s)$.
Denote
\begin{equation}
      z(t) = \int_0^t S(t-s) \Phi dB^H (s).
\end{equation}
Then $z$, if it is well defined,
is the unique mild solution of the following linear stochastic evolution equation
\begin{equation}
    dz(t) = Az(t)dt + \Phi dB^H(t), \quad z(0)=0 \in H.
\end{equation}
%As mention in the next section,
In order to obtain the solution of non-Newtonian systems driven by fBm
we only need to guarantee the existence and regularity of the stochastic Wiener-type convolution
(see the proof of Theorem \ref{Theorem: existence of solu} ).
There are three groups of assumption on stochastic convolution:
%$    \int_0^t S(t-s) \Phi dB^H(s)
%     = \sum_{i=1}^{\infty}
%       \sqrt{\lambda_i}  \int_0^t  S(t-s)\Phi e_i d \beta_i^H (s)$:
\begin{enumerate}
  \item[(A1)]
     $Q \in \mathcal{L}_1(H)$,
     \ $\Phi \equiv id_H$;
%     $ |S(t)|_{\mathcal{L}(V)} \leq C e^{-at}$;
  \item[(A2)]
    $Q \equiv id_H$,
    \qquad $\Phi \in \mathcal{L}_2(H)$;
%    or $ |S(t) \Phi |_{\mathcal{L}_2(H)} < ct^{-\gamma} $ for $\gamma \in [0,H)$;
  \item[(A3)]
    $Q \equiv id_H$,
    \qquad $\Phi \in \mathcal{L}(H)$ such that $ \Phi \Phi^* \in \mathcal{L}_1(H)$.
\end{enumerate}
Here we use notation: $\mathcal{L}_1(H)$ the space of all nuclear operators on $H$;
$\mathcal{L}_2(H)$ the space of all Hilbert-Schmidt operators on $H$
(for detailed see Appendix C in \cite{DaPrato Zabczyk stochastic}).
These assumptions come from Maslowski, Schmalfuss \cite{Maslowski Schmalfus}
and Duncan, Maslowski, Pasik-Duncan \cite{Duncan Maslowski PasikDuncan}
and Tindel, Tudor, Veins \cite{Tindel Tudor Viens} respectively.
We first give a general consequence about the existence and regularity of stochastic convolution
under each of these assumptions,
then remark it.

\begin{proposition}\label{Proposition: exis an regu of stoch convol}
  Under each of assumption (A1)-(A4), the stochastic convolution
  $ z(t) = \int_0^t S(t-s) \Phi dB^H (s)$ is well defined and $z \in C((0,T];V)$.
\end{proposition}

\begin{proof}
  The assumption (A4) and the corresponding proof we will state later.\\
  Under assumption (A1):\\
  By setting $V=H$, $E=V$ and $\Phi \equiv id_H$ in \cite{Maslowski Schmalfus},
  the infinite dimensional fBm $B^H(t)=\sum^{\infty}_{i=1} \sqrt{\widetilde{\lambda}_i} e_i \beta_i^H(t)$
  is indeed a $H$-valued process with covariance $Q$ since $Q \in \mathcal{L}_1(H)$.
  In order to apply the theory in \cite{Maslowski Schmalfus} we only need to verify that
  the semigroup $S$ is analytic and exponentially stable, that is,
  \begin{equation}
    |S(t)|_{\mathcal{L}(V)} \leq C e^{-at}.
  \end{equation}
  And this is straightforward due to the properties of operator $A$.
  Thus we can use Proposition 3.1. in \cite{Maslowski Schmalfus} to obtain that
  the stochastic convolution $z$ is well defined by the variation of constants formula
  and it has a $V$-continuous modification.

  Under assumption (A2):\\
  Sometimes the assumption that $Q$ is nuclear is not convenient.
  In this case we consider the genuine cylindrical fBm
  (i.e. $Q \equiv id_H$).
  And this standard cylindrical fBm can be represented by the formal series
  \begin{equation}
    B^H(t)=\sum^{\infty}_{i=1} \beta_i^H(t) e_i
  \end{equation}
  that does not converge a.s. in $H$.
  Since $\Phi \in \mathcal{L}_2(H)$ and $S$ is an analytic semigroup,
  we can use Proposition 2.6 in \cite{Duncan Maslowski PasikDuncan} to obtain a stronger conclusion.
  That is, let $\alpha < H- \frac{1}{2}$
  then the stochastic convolution $z$ has a $C^{\alpha}([0,T];D(A^{\frac{1}{2}}))$ version.
  In particular, there is a $C([0,T];V)$ version.

  Under assumption (A3):\\
  Sometimes the assumption that $\Phi$ is Hilbert-Schmidt is not necessary.
  We can relax it to the case $\Phi \Phi^* \in \mathcal{L}_1(H)$.
  Since $A^{-2H} \in \mathcal{L}(H)$, we have
  \begin{equation}
          tr(\Phi^* A^{-2H} \Phi)
           = tr( A^{-2H} \Phi\Phi^*)
           \leq |A^{-2H}|_{\mathcal{L}(H)} \cdot |\Phi\Phi^*|_{\mathcal{L}_1(H)}
           < \infty.
  \end{equation}
  Thus, the assumptions in Theorem 1 \cite{Tindel Tudor Viens} are fulfilled and
  we can deduce that $z$ is well-defined and belongs to $L^2(\Omega,H)$.
  For the regularity of $z$, we need to check the condition of Theorem 4 in \cite{Tindel Tudor Viens}.
  Since $H > \frac{1}{2}$, for all $\alpha \in (\frac{1}{2},H)$ we have
  \begin{equation}
          tr(\Phi^* A^{-2(H-\alpha)} \Phi)
           = tr( A^{-2(H-\alpha)} \Phi\Phi^*)
           \leq |A^{-2(H-\alpha)}|_{\mathcal{L}(H)} \cdot |\Phi\Phi^*|_{\mathcal{L}_1(H)}
           < \infty.
  \end{equation}
  Thus, we have $z \in C^{\alpha - \frac{1}{2}}([0,T];D(A^{\frac{1}{2}}))$.
  In particular, $z \in C([0,T];V)$.
\end{proof}

\begin{remark}
  No matter which assumption holds, the key point is the concept of compactness
  which guarantee us to handel the infinite-dimensional problem in a finite-dimensional manner.
  Actually, the nuclear operators (elements of $\mathcal{L}_1(H)$) are compact.
  The Hilbert-Schmidt operators (elements of $\mathcal{L}_2(H)$) are compact too.
\end{remark}

Next we discuss the fundamental example in which
the boundary is square and the parameter $Q \equiv \Phi \equiv id_H$
(!they are not compact). That is
\begin{enumerate}
  \item[(A4)] $ \mathcal{O}  = [-\pi,\pi] \times [-\pi,\pi]  $, \quad
  $Q \equiv \Phi \equiv id_H$.
\end{enumerate}
Before proving the existence and regularity of stochastic convolution $z$,
we state the following lemma in \cite{Li Huang}
(which is based on \cite{Kelliher}) about the spectrum of operator $A$.

\begin{lemma}
  \cite{Li Huang} Under square domain, the eigenvalues of operator A satisfy
\begin{equation}\label{3}
    \lambda_{mn} \geq (m^2 + n^2)^2 , \quad m,n \in \mathbb{N}.
\end{equation}
\end{lemma}

Now we can finish the proof of Proposition \ref{Proposition: exis an regu of stoch convol}.
\begin{proof}[proof of Proposition \ref{Proposition: exis an regu of stoch convol} under assumption (A4)]
  The existence part is based on Theorem 1 in \cite{Tindel Tudor Viens}.
  Let us estimate the mean square of $z$.
  \begin{equation}
    \begin{split}
      I_t : & = \mathbb{E} \left| \int_0^t S(t-s)dB^H(s) \right|_H^2
        = \mathbb{E} \left| \sum_{i=1}^{\infty} \int_0^t S(t-s) e_i d\beta^H_i(s) \right|_H^2 \\
       & =  \sum_{i} \int_0^t \int_0^t  <S(t-u) e_i, S(t-v)e_i >_H |u-v|^{2H-2} du dv  \\
       & =  \sum_{i} \int_0^t \int_0^t  e^{-(2t-u-v)\lambda_i} |u-v|^{2H-2} du dv  \\
       & = 2 \sum_{i} \int_0^t \int_0^u  e^{-(2t-u-v)\lambda_i} (u-v)^{2H-2}  dv du
    \end{split}
  \end{equation}
  By the change of variable $x=u-v$ and $y= \lambda_i x$, we get
  \begin{equation}
    \begin{split}
        I_t & = 2\sum_{i} \int_0^t \int_0^u  e^{-(2t-2u+x)\lambda_i} x^{2H-2}  dx du \\
         & = 2\sum_{i} \lambda_i^{1-2H} e^{-2 \lambda_i t}
               \int_0^t e^{2 \lambda_i u} \int_0^{\lambda_i u} y^{2H-2} e^{-y}  dy du \\
         & = 2\sum_{i} \lambda_i^{1-2H} e^{-2 \lambda_i t}
               \int_0^{\lambda_i t} y^{2H-2} e^{-y} \int_{\lambda_i^{-1}y}^{t} e^{2 \lambda_i u} du dy \\
         & =  \sum_{i} \lambda_i^{-2H}
               \int_0^{\lambda_i t} y^{2H-2} e^{-y} (1-e^{-2(\lambda_i t -y)}) dy \\
         & \leq  \sum_{i} \lambda_i^{-2H}
               \int_0^{\infty} y^{2H-2} e^{-y} dy \\
         & \leq  \Gamma(2H-1) \sum_{i,j=1}^{\infty} \frac{1}{(i^2 + j^2)^{4H} } \\
         & = 2 \Gamma(2H-1) \cdot \beta_D(4H) \cdot \xi(4H) < \infty,
    \end{split}
  \end{equation}
  where $\Gamma(s)$ is Gamma function, $\beta_D(s)$ is the Dirichlet beta function and
    $\zeta(s)$ is the Riemann zeta function (for definition see \cite{Borwein Borwein}).
    This yields the existence of $z$.
    Since the proof of regularity is a modification of Lemma 5.13 in \cite{DaPrato Zabczyk} and
Proposition 3.1 in \cite{Maslowski Schmalfus}, we omit it.
\end{proof}

A fundamental concept in the theory of random dynamical system is the notion of
metric dynamical system.
It is a model for a noise which is the source of perturbation of a dynamical system.
We now recall the notions of RDS. For details we refer to \cite{Crauel Flandoli}.
\begin{definition}
  Let $(\Omega, \mathcal{F},P)$ be a probability space.
  A measurable flow $\theta = \{ \theta_t \}_{t \in \real}$ on $\Omega$
  is defined as a mapping
  \begin{equation}
    \theta : \real \times \Omega \rightarrow \Omega
  \end{equation}
  which is $\mathcal{B}(\real) \times \mathcal{F};\mathcal{F}$-measurable and satisfies
  the flow property
  \begin{equation}
    \theta_t \theta_{\tau} = \theta_{t+\tau} \quad \text{for }t,\tau \in \real
  \end{equation}
  and $\theta_0=id_{\Omega}$.
  In addition, we suppose that the measure $P$ is invariant with respect to the flow $\theta$.
  Then the quadruple $(\Omega, \mathcal{F},P,\theta)$ is called a metric dynamical system (MDS).
\end{definition}

Due to the property of stationary increments, we can switch fBm to the equivalent canonical realization.
Let $\Omega = C(\real,V)$ be the set continuous functions defined on $\real$ and
with values $\real$ such that $\omega(0)=0$ for $\omega \in \Omega$.
$\Omega$ being equipped with the compact open topology and
let $\mathcal{F}$ be the associated Borel-$\sigma$-algebra (!incomplete).
The operators $\theta_t$ forming the flow are given by the Wiener shift:
\begin{equation}
    \theta_t \omega (\cdot) = \omega (\cdot +t)-\omega(t), \quad t \in \real.
\end{equation}

\begin{definition}
  Let $E$ be a complete and separable metric space.
  A random dynamical system (RDS) with space $E$ carried by a metric dynamical system
  $(\Omega, \mathcal{F},P,\theta)$ is given by mapping
  \begin{equation}
    \varphi : \real_+ \times \Omega \times E \rightarrow E
  \end{equation}
  which is $\mathcal{B}(\real_+) \times \mathcal{F} \times \mathcal{B}(E);\mathcal{B}(E)$-measurable
  and possess the cocycle property:
  \begin{align}
    \varphi (t+\tau,\omega,x) &= \varphi (\tau,\theta_t \omega, \varphi(t,\omega,x))
     \quad \text{for } t,\tau \in \real_+ \text{ and } x\in E, \omega \in \Omega
     \label{Equation: cocycle property} \\
    \varphi(0,\omega, \cdot) &= id_{\Omega}.
  \end{align}
\end{definition}
In order to generate the RDS from \eqref{Equation: integral form},
we need to check the measurability and cocycle property for all $\omega \in \Omega$.
Since the integrands of Wiener-type stochastic integral \eqref{Equation: integral form}
is deterministic, one can solve the equation in a pathwise way.
Thus, the cocycle property \eqref{Equation: cocycle property} can be satisfied for any noise path
$\omega \in \Omega$ by the uniqueness of the solution map.

\begin{definition}
  \begin{enumerate}
    \item[(i)] A set valued map $K: \Omega \rightarrow 2^E$ taking value in the closed subsets of $E$
        is said to be measurable if for each $x \in E$ the map $\omega \mapsto d(x,K(\omega))$ is measurable,
        where
        \begin{equation}
            d(A,B) = \sup_{x\in A} \inf_{y \in B} d(x,y).
        \end{equation}
         A set valued measurable map $K: \Omega \rightarrow 2^E$ is called a random set.
    \item[(ii)] Let $A$,$B$ be random sets. $A$ is said to attract $B$ if
    \begin{equation}
        d(\varphi(t, \theta_{-t} \omega)B(\theta_{-t} \omega),A(\omega)) \rightarrow 0, \text{ as }
        t \rightarrow \infty \text{ P-a.s.}
    \end{equation}
    $A$ is said to absorb $B$ if P-a.s. there exists an absorption time $t_B(\omega)$
    such that for all $t \geq t_B(\omega)$
    \begin{equation}
        \varphi(t,\theta_{-t}\omega)B(\theta_{-t}\omega) \subset A(\omega).
    \end{equation}
    \item[(iii)] The $\Omega$-limit set of a random set $K$ is defined by
    \begin{equation}
        \Omega_K(\omega)=
          \bigcap_{T\geq 0} \overline{\bigcup_{t\geq T} \varphi(t,\theta_{-t}\omega) A(\theta_{-t} \omega)}.
    \end{equation}
  \end{enumerate}
\end{definition}

\begin{definition}
  A random attractor for an RDS $\varphi$ is a compact random set $A$ satisfying P-a.s.:
  \begin{enumerate}
    \item[(i)] $A$ is invariant, i.e. $\varphi(t,\omega) A(\omega) = A(\theta_t \omega)$ for all $t>0$.
    \item[(ii)] A attracts all deterministic bounded sets $B \subset E$.
  \end{enumerate}
\end{definition}

The following proposition (cf. \cite{Crauel Flandoli} Theorem 3.11) yields a sufficient criterion
for the existence of a random attractor.

\begin{proposition}\label{Proposition: exi rand attr}
  (\cite{Crauel Flandoli}) Let $\varphi$ be an RDS and assume the existence of a compact random set $K$
  absorbing every deterministic bounded set $B \subset E$.
  Then there exists a random attractor $A$, given by
  \begin{equation}
    A(\omega) = \overline{ \bigcup_{B \subset E,\ B \ bounded} \Omega_B(\omega)}
  \end{equation}
\end{proposition}

\section{Solution of non-Newtonian fluids}

We seek the solution of 2D stochastic non-Newtonian fluids by the modified fixed point theorem
in the space $X=C([0,T];H)\cap L^2(0,T;V)$.
First we prove the local existence and uniqueness results.
For $ u \in X$, let
\begin{align}
  J_1(u): &= - \int_0^{\cdot} S(\cdot-s)B(u(s))ds , \\
  J_2(u): &= - \int_0^{\cdot} S(\cdot-s)N(u(s))ds .
\end{align}
We have following estimates.

\begin{lemma}\label{Lemma: estimate J1}
  $J_1 : X \rightarrow X $ and
  for all $ u,v \in X$, we have
    \begin{align}
      |J_1(u)|^2_X & \leq c_1 |u|_X^4 ,      \label{Estimate: J1 1}\\
      \begin{split}
        |J_1(u)-J_1(v)|^2_X & \leq c_2 \left( |u|^2_{C([0,T];H)}\cdot |u|^2_{L^2(0,T);V}
       + |v|_{C([0,T];H)}^2 \cdot |v|^2_{L^2(0,T);V} \right)^{\frac{1}{2}}  \label{Estimate: J1 2}\\
       & \quad \cdot |u-v|^2_X.
      \end{split}
    \end{align}
\end{lemma}

\begin{proof}
  According to \cite{Bloom Hao existence} Lemma 2.6, we have
  \begin{equation}
      B(u) \in L^2(0,T;V') \quad \forall u \in X.
  \end{equation}
  Then $J_1(u)$ is a weak solution of the following linear differential equation
  \begin{align}\label{Equation: linear for J1}
    & \frac{dJ(t)}{dt} + AJ(t)+B(u(t))=0, \quad t \in [0,T], \\
    & J(0)=0,
  \end{align}
  and $J_1 \in C([0,T];H)\cap L^2(0,T;V) = X$.
  Thus $J_1$ maps $X$ into $X$.
  Since $J$ is a weak solution of above equation,
  we can take inner product of $J$ in \eqref{Equation: linear for J1} and obtain
  \begin{equation}
    \begin{split}
      \frac{1}{2} \frac{|J(t)|^2}{dt} + |J(t)|^2_V & = -<B(u(t)),J(t)>  \\
      & \leq |B(u(t))|_{V'} \cdot |J(t)|_V \\
      & \leq \frac{1}{2} |B(u(t))|_{V'}^2 + \frac{1}{2} |J(t)|_V^2.
    \end{split}
  \end{equation}
  Integrating with respect to $t$ over $[0,t]$, we get
  \begin{equation}
    |J(t)|^2 + \int_0^t |J(s)|_V^2 ds \leq \int_0^t |B(u(s))|^2_{V'} ds.
  \end{equation}
  Since
  \begin{equation}
    \begin{split}
      \int_0^T |B(u(t))|_{V'}^2 dt
      & \leq c_1 \int_0^T |u(t)|^2 \cdot |u(t)|_V^2 dt \\
      & \leq c_1 \cdot |u|^2_{C([0,T];H)} \cdot \int_0^T |u(t)|_V^2 dt \\
      & \leq \frac{c_1}{2} \left( |u|^4_{C([0,T];H)} + |u|^4_{L^2(0,T;V)} \right) \\
      & \leq \frac{c_1}{2}  |u|_X^4,
    \end{split}
  \end{equation}
  we have
  \begin{equation}
    |J|_X^2 \leq 2 \left( |J|^2_{C([0,T];H)}+|J|^2_{L^2(0,T;V)} \right)
    \leq c_1 |u|_X^4.
  \end{equation}
  Next we prove \eqref{Estimate: J1 2}.
  For $u,v \in X$,
  let $w=J_1(u)-J_1(v)$. Then $w$ is the weak solution of
  \begin{align}
    & \frac{dw(t)}{dt} + Aw(t) + B(u(t))-B(v(t))=0, \\
    & w(0)=0.
  \end{align}
  By the same argument we have
  \begin{equation}
    |w(t)|^2 + \int_0^t|w(s)|_V^2 ds \leq \int_0^t |B(u(s))-B(v(s))|_{V'}^2 ds.
  \end{equation}
  Hereafter we estimate $|B(u)-B(v)|_{V'}$.
  For all $ \phi \in V$,
  \begin{equation}
    \begin{split}
      & \quad \ |<B(u)-B(v), \phi >|   \\
      & = |b(u,u,\phi)-b(v,v,\phi)| \\
      & \leq |b(u-v,y,\phi)|+|b(v,u-v,\phi)| \\
      & \leq C \left( |u-v|^{\frac{1}{2}} |u-v|^{\frac{1}{2}}_V
       \cdot |\phi|_V \cdot
       |u|^{\frac{1}{2}} |u|^{\frac{1}{2}}_V
       + |v|^{\frac{1}{2}} |v|^{\frac{1}{2}}_V
       \cdot |\phi|_V \cdot
       |u-v|^{\frac{1}{2}} |u-v|^{\frac{1}{2}}_V \right)  \\
      & = C \left(  |u|^{\frac{1}{2}} |u|^{\frac{1}{2}}_V
       + |v|^{\frac{1}{2}} |v|^{\frac{1}{2}}_V  \right)
        |u-v|^{\frac{1}{2}} |u-v|^{\frac{1}{2}}_V  |\phi|_V.
    \end{split}
  \end{equation}
  Therefore,
  \begin{equation}
    |B(u)-B(v)|_{V'}
    \leq c_2 \left(  |u|^{\frac{1}{2}} |u|^{\frac{1}{2}}_V
       + |v|^{\frac{1}{2}} |v|^{\frac{1}{2}}_V  \right)
        |u-v|^{\frac{1}{2}} |u-v|^{\frac{1}{2}}_V.
  \end{equation}
  Finally,
  \begin{equation}
    \begin{split}
      & \quad  |J_1(u)-J_1(v)|^2_X \\
      & \leq  2 \sup_{t\in [0,T]}|w(t)|^2 + 2 \int_0^T|w(s)|_V^2 ds \\
      & \leq  2C \int_0^T \left(  |u(s)|^{\frac{1}{2}} |u(s)|^{\frac{1}{2}}_V
        + |v(s)|^{\frac{1}{2}} |v(s)|^{\frac{1}{2}}_V  \right)^2
        |u(s)-v(s)| \cdot |u(s)-v(s)|_V ds \\
      & \leq C \left( \int_0^T \left(  |u(s)|^{\frac{1}{2}} |u(s)|^{\frac{1}{2}}_V
        + |v(s)|^{\frac{1}{2}} |v(s)|^{\frac{1}{2}}_V  \right)^4
        |u(s)-v(s)|^2  ds
        + \int_0^T |u(s)-v(s)|_V^2 ds \right) \\
      & \leq C \left( |u-v|_{C([0,T];H)}^2
      \int_0^T \left(  |u(s)|^{\frac{1}{2}} |u(s)|^{\frac{1}{2}}_V
        + |v(s)|^{\frac{1}{2}} |v(s)|^{\frac{1}{2}}_V  \right)^4
         ds
        + |u-v|_{L^2(0,T;V)}^2 ds \right) \\
      & \leq C \left( 4|u-v|_{C([0,T];H)}^2
      \int_0^T \left(  |u(s)|^2 |u(s)|^2_V
        + |v(s)|^2 |v(s)|^2_V  \right)
         ds
        + |u-v|_{L^2(0,T;V)}^2 ds \right) \\
      & \leq C \left( 4|u-v|_{C([0,T];H)}^2
      \left(  |u|^2_{C([0,T];H)} |u|^2_{L^2(0,T;V)}
        + |v|^2_{C([0,T];H)} |v|^2_{L^2(0,T;V)}  \right)
        + |u-v|_{L^2(0,T;V)}^2 ds \right) \\
      & \leq 2C
      \left(  |u|^2_{C([0,T];H)} |u|^2_{L^2(0,T;V)}
        + |v|^2_{C([0,T];H)} |v|^2_{L^2(0,T;V)}  \right)^{\frac{1}{2}}
        \left( |u-v|_{C([0,T];H)}^2+ |u-v|_{L^2(0,T;V)}^2 ds \right) \\
      & \leq c_2
      \left(  |u|^2_{C([0,T];H)} |u|^2_{L^2(0,T;V)}
        + |v|^2_{C([0,T];H)} |v|^2_{L^2(0,T;V)}  \right)^{\frac{1}{2}}
        |u-v|_X^2.
    \end{split}
  \end{equation}
\end{proof}

\begin{lemma}\label{Lemma: estimate J2}
  $J_2 : X \rightarrow X $ and
  for all $ u,v \in X$, we have
    \begin{align}
      |J_2(u)|^2_X        & \leq  c_3 |u|^2_{L^2(0,T;V)}   ,   \label{Estimate: J2 1}\\
      |J_2(u)-J_2(v)|^2_X & \leq  c_4 T |u-v|^2_X          .   \label{Estimate: J2 2}
    \end{align}
\end{lemma}

\begin{proof}
  For all $ u \in X$,
  we have $N(u)\in L^2(0,T;V)$ (see \cite{Bloom Hao existence} Lemma 2.6).
  As the proof of Lemma \ref{Lemma: estimate J1},
  we can prove that $J_2$ maps $X$ into $X$,
  $J_2(u)$ is the weak solution of
  \begin{align}\label{Equation: linear for J2}
    & \frac{dJ(t)}{dt} + AJ(t)+N(u(t))=0, \quad t \in [0,T], \\
    & J(0)=0
  \end{align}
  and the following estimate holds:
  \begin{equation}
    |J(t)|^2 + \int_0^t |J(s)|_V^2 ds \leq \int_0^t |N(u(s))|^2_{V'} ds .
  \end{equation}
  By \cite{Bloom Hao existence} Lemma 2.6, we have
  \begin{equation}
    |N(u)|_{V'} \leq C |u|_V.
  \end{equation}
  Therefore,
  \begin{equation}
      |J_2(u)|^2_X
      \leq 2 \int_0^t |N(u(s))|^2_{V'} ds
      \leq 2 C \int_0^t |u(s)|^2_{V} ds
      \leq c_3 |u|^2_{L^2(0,T;V)} .
  \end{equation}
  Next we prove \eqref{Estimate: J2 2}.
  For $ u,v \in X$,
  let $w=J_2(u)-J_2(v)$. Then $w$ is the weak solution of
  \begin{align}
    & \frac{dw(t)}{dt} + Aw(t) + N(u(t))-N(v(t))=0, \\
    & w(0)=0.
  \end{align}
  Similar to Lemma \ref{Lemma: estimate J1}, we have
  \begin{equation}
    |w(t)|^2 + \int_0^t|w(s)|_V^2 ds \leq \int_0^t |N(u(s))-N(v(s))|_{V'}^2 ds.
  \end{equation}
  For all $ \alpha \in (0,1)$ and $\phi \in V$, we have
  \begin{equation}
    \begin{split}
      & \quad \ |<N(u)-N(v),\phi>| \\
      & = 2|\int_{\mathcal{O}} \left( \mu (u) e_{ij}(u)-\mu (v) e_{ij}(v) \right)
       e_{ij}(\phi) dx|
    \end{split}
  \end{equation}
  Inspired by the technique in \cite{Zhao Zhou} Lemma 3.1, we set
  \begin{equation}
    F(s)=2 \mu_0 (\epsilon + |s|^2)^{-\alpha /2}s,
  \end{equation}
  where
  \begin{equation}
    s=
    \begin{pmatrix}
      s_1 & s_2 \\
      s_3 & s_4
    \end{pmatrix}
    \in \mathbb{R}^4,
    \quad |s|^2=\sum_{i=1}^4 s_i^2, \quad s_i\in \mathbb{R},i=1,2,3,4.
  \end{equation}
  Then the first order Fr\'echet derivative of $F(s)$ is
  \begin{equation}
    DF(s)=2\mu_0(\epsilon +|s|^2)^{-\alpha /2}
    \begin{pmatrix}
      1-\frac{\alpha s_1^2}{\epsilon + |s|^2}  &  -\frac{\alpha s_1 s_2}{\epsilon + |s|^2} &
      -\frac{\alpha s_1 s_3}{\epsilon + |s|^2} & -\frac{\alpha s_1 s_4}{\epsilon + |s|^2} \\
      -\frac{\alpha s_1 s_2}{\epsilon + |s|^2} &  1-\frac{\alpha s_2^2}{\epsilon + |s|^2} &
      -\frac{\alpha s_2 s_3}{\epsilon + |s|^2} & -\frac{\alpha s_2 s_4}{\epsilon + |s|^2} \\
      -\frac{\alpha s_1 s_3}{\epsilon + |s|^2} & -\frac{\alpha s_2 s_3}{\epsilon + |s|^2} &
      1-\frac{\alpha s_3^2}{\epsilon + |s|^2}  & -\frac{\alpha s_3 s_4}{\epsilon + |s|^2} \\
      -\frac{\alpha s_1 s_4}{\epsilon + |s|^2} & -\frac{\alpha s_2 s_4}{\epsilon + |s|^2} &
      -\frac{\alpha s_3 s_4}{\epsilon + |s|^2} &  1-\frac{\alpha s_4^2}{\epsilon + |s|^2}  .
    \end{pmatrix}
  \end{equation}
  Since $0<\alpha<1$, we have
  \begin{equation}
    \left| -\frac{\alpha s_i s_j}{\epsilon + |s|^2} \right|
    < \left| -\frac{s_i s_j}{\epsilon + |s|^2} \right|
    < \frac{1}{\epsilon}, \quad i,j=1,2,3,4,
  \end{equation}
  and
  \begin{equation}
    0<1- \frac{\alpha s_i^2}{\epsilon +|s|^2} <1, \quad i=1,2,3,4.
  \end{equation}
  Consequently,
  \begin{equation}
    |DF(s)| \leq 2\mu_0(\epsilon +|s|^2)^{-\alpha /2} \sqrt{4+\frac{12}{\epsilon^2}},
    \quad \forall s \in \mathbb{R}^4.
  \end{equation}
  Similarly, the second order Fr\'echet derivative of $F(s)$ is a three-dimensional matrix
  \begin{equation}
    D^2 F(s)= \left( \frac{\partial^2 F_i (s)}{\partial s_j \partial s_k} \right),
    \quad i,j,k=1,2,3,4
  \end{equation}
  where $F_i(s)=2\mu_0 (\epsilon +|s|^2)^{-\alpha /2} s_i$.
  By some computation we see that
  \begin{equation}
    |DF(s)| + |D^2 F(s)| \leq C (\mu_0, \epsilon, \alpha)
    \quad \forall s_i \in \mathbb{R}, i=1,2,3,4
  \end{equation}
  where $c_4$ is a positive constant depending on $\mu_0,\epsilon$ and $\alpha$.
  For all $a,b \in \mathbb{R}^4$,
  \begin{equation}
    F(b)-F(a)= \int_0^1 DF \left( a+\tau(b-a) \right) (b-a) d\tau.
  \end{equation}
  Taking $a=e(u)=(e_{ij}(u)),b=e(v)=(e_{ij}(v))$,
  applying the integration by parts and the above inequality about $F(s)$, we have
  \begin{equation}
    \begin{split}
     <N(u)-N(v),\phi>
     & = \left( F(e(u))-F(e(v)), \nabla \phi  \right)  \\
     & = \int_{\mathcal{O}} |DF(\xi)| \cdot |e(u)-e(v)|\cdot |\nabla \phi| dx \\
     & \leq 2 \mu_0 \epsilon^{-\frac{\alpha}{2}} \sqrt{4+\frac{12}{\epsilon^2}}
       \int_{\mathcal{O}} |e(u-v)| \cdot |\nabla \phi| dx \\
     & \leq C |e(u-v)| \cdot |\nabla \phi| \\
     & \leq C |u-v|_{H_0^1} \cdot |\phi|_{H_0^1} \\
     & \leq C |u-v|^{1/2} \cdot |u-v|_{H^2}^{1/2} \cdot|\phi|_{H_0^1}
       \quad \text{(by the Sobolev interpolation theorem)} .\\
    \end{split}
  \end{equation}
  Finally,
  \begin{equation}
    \begin{split}
     |J_2(u)-J_2(v)|^2_X
    & \leq 2 \int_0^T |N(u(s))-N(v(s))|_{V'}^2 ds \\
    & \leq C  \int_0^T  |u(s)-v(s)| \cdot |u(s)-v(s)|_{H^2} ds \\
    & \leq C \cdot |u-v|_{C([0,T];H)} \cdot |u-v|_{L^1(0,T;V)} \\
    & \leq C \cdot T \cdot |u-v|_{C([0,T];H)} \cdot |u-v|_{L^2(0,T;V)} \\
    & \leq c_4 T |u-v|_X^2  .   \\
    \end{split}
  \end{equation}
\end{proof}

We apply the following version of the contraction mapping theorem.

\begin{lemma}\label{Lemma: fix point}
  (\cite{DaPrato Zabczyk} Lemma 15.2.6)
  Let $F$ be a transformation from a Banach space $E$ into $E$,
  $\phi \in E$ and $M>0$ a positive number.
  If $F(0)=0$, $| \phi |_E \leq \frac{1}{2}M$ and
  \begin{equation}
    |F(u)-F(v)|_E \leq \frac{1}{2} |u-v|_E \quad \forall u,v \in B_E(M),
  \end{equation}
  then the equation
  \begin{equation}
    u = \phi +F(u)
  \end{equation}
  has a unique solution $u \in E$ satisfying $u \in B_E(M)$.
\end{lemma}

We now prove the main result of this paper.

\begin{theorem}\label{Theorem: existence of solu}
  (local existence and uniqueness of solution) For all $u_0 \in H$, there exists $T_0>0$ s.t.
  equation \eqref{Problem: differential form} admits a unique solution
  $u \in C([0,T_0];H) \cap L^2(0,T_0;V)$ for all $\omega \in \Omega$
  in the sense of \eqref{Equation: integral form}.
\end{theorem}

\begin{proof}
  Fix $\omega \in \Omega$. Let
  \begin{equation}
    \phi (t)= S(t)u_0 + z(t).
  \end{equation}
  By the property of analytic semigroup and
  Lemma \ref{Proposition: exis an regu of stoch convol},
  $S(\cdot)u_0,z \in C([0,T];V) \subset X$.
  Then we have
  \begin{equation}
    |\phi |_X \leq |S(\cdot)u_0|_X+|z|_X \leq 2|u_0|+|z|_X.
  \end{equation}
  Let $M(\omega)=2(2|u_0|+|z(\omega)|_X)$.
  Construct the mapping $\mathcal{F}=J_1+J_2$,
  then for all $u,v \in X$, we have
  \begin{equation}
   \begin{split}
    & \quad |\mathcal{F}(u)-\mathcal{F}(v)|_X  \\
    & \leq |J_1(u)-J_1(v)|_X + |J_2(u)-J_2(v)|_X \\
    & \leq  c_2^{\frac{1}{2}}  \left( |u|^2_{C([0,T];H)}\cdot |u|^2_{L^2(0,T);V}
       + |v|_{C([0,T];H)}^2 \cdot |v|^2_{L^2(0,T);V} \right)^{\frac{1}{4}}
      \cdot |u-v|_X  \\
    & \quad +  (c_4  T)^{\frac{1}{2}} |u-v|_X \qquad
       \text {(by Lemma \ref{Lemma: estimate J1} and \ref{Lemma: estimate J2})}   \\
    & \leq  (c_2  M)^{\frac{1}{2}}  \left( |u|^2_{L^2(0,T);V}
       +   |v|^2_{L^2(0,T);V} \right)^{\frac{1}{4}}  |u-v|_X
       +   (c_4 T)^{\frac{1}{2}} |u-v|_X .
   \end{split}
  \end{equation}
  Due to the absolute continuity property of Bochner integral,
  we can choose $ \tau \in (0,1]$ s.t.
  \begin{equation}
    \left(|u|^2_{L^2(0,\tau);V} + |v|^2_{L^2(0,\tau);V} \right)^{\frac{1}{4}}
    \leq (2Mc_2)^{-\frac{1}{2}}.
  \end{equation}
  Let $T_0 = \min \{\tau, 1, \frac{1}{16c_4}  \}$ and $X_{T_0}:=C([0,T_0];H) \cap L^2(0,T_0;V)$.
  We have
  \begin{equation}
     |\mathcal{F}(u)-\mathcal{F}(v)|_{X_{T_0}}
     \leq (\frac{1}{4}+\frac{1}{4})|u-v|_{X_{T_0}} = \frac{1}{2}|u-v|_{X_{T_0}}.
  \end{equation}
  Applying the modified fixed point lemma \ref{Lemma: fix point},
  equation
  \begin{equation}
    u= \phi + \mathcal{F}(u) \equiv S(\cdot)u_0 + z + J_1(u)+J_2(u)
  \end{equation}
  has a unique solution $u$  in $C([0,T_0];H) \cap L^2(0,T_0;V)$
  and the solution satisfies $|u|_{X_{T_0}} \leq M$.
\end{proof}

In the rest part of this section we obtain a priori estimates and global existence.
Denote by $u$, the local solution of \eqref{Equation: integral form} over $[0,T_0]$.
Let $v(t)=u(t)-z(t)$.
Then $v(t)$ is the mild solution of equation
\begin{equation}\label{Equation: integral form for v}
    v(t)=S(t)u_0 - \int_0^t S(t-s)B(v(s)+z(s))ds - \int_0^t S(t-s)N(v(s)+z(s))ds.
\end{equation}
Therefore, $v(t)$ is the weak solution of the following differential equation with random parameters:
\begin{align}\label{Equation: differential form for v}
  & \frac{v(t)}{dt}+Av(t)+B(v(t)+z(t))+N(v(t)+z(t))=0, \\
  & v(0)=u_0.
\end{align}
Inspired by \cite{DaPrato Zabczyk} Chapter 15.3,
we give a priori estimate which ensures the global existence of solution.

\begin{proposition}\label{Proposition: Extention of solu}
  Assume that $v$ is the solution of \eqref{Equation: integral form for v} on the interval.
  Then we have
\begin{align}
  \begin{split}\label{Equation: estimate v in CH}
    \sup_{t \in [0,T]} |v(t)|^2
    \leq e^{c_5 \int_0^T |z(s)|_{H_0^1}^2 ds} |u_0|^2
    + \int_0^T e^{c_5 \int_s^T |z(r)|_{H_0^1}^2 dr} g_1(s) ds,
  \end{split} \\
  \begin{split}\label{Equation: estimate v in LV}
    \int_0^T |v(t)|_V^2 dt
    \leq c_6 |u_0|^2
    + c_5 c_6  \sup_{t \in [0,T]} |v(t)|^2 \int_0^T |z(s)|_{H_0^1}^2 ds
    + c_6 \int_0^T g_1(s) ds
  \end{split}
\end{align}
  where $c_5$ and $c_6$ are positive constants depending on $\lambda_1$ and $\mathcal{O}$,
  $g_1$ is an integrable function depending on $z$.
\end{proposition}

\begin{proof}
Multiple \eqref{Equation: differential form for v} by $v(t)$
and then integrate over $\mathcal{O}$. We have
\begin{equation}
  \begin{split}
    \frac{1}{2} \frac{d|v(t)|^2}{dt} +|v(t)|^2_V
    & = - b\left(v(t)+z(t),v(t)+z(t),v(t)\right)  - <N(v(t)+z(t)),v(t)> \\
    &\leq |b \left( v(t)+z(t),z(t),v(t)+z(t)\right)| - <N(z(t)),v(t)>.
  \end{split}
\end{equation}
The above inequality take advantage of $<N(v),v> \ \geq 0$ (see \cite{Bloom Hao existence})
and orthogonality property of $b$ (see \cite{RogerTemam} (2.21)).
In the sequel we omit the time variable $t$.
Firstly we estimate trilinear form $b$.
\begin{equation}
    \begin{split}
      & \quad b\left(v+z,z,v+z\right) \\
      & \leq C_1 |v+z| \cdot
        |z|_{H_0^1} \cdot
        |v+z|_{H_0^1}  \\
      & \leq \frac{C_1}{2C_2} |z|^2_{H_0^1} \cdot |v+z|^2
        + \frac{C_1 C_2}{2} |v+z|^2_{H_0^1} \\
      & \leq \frac{C_1}{C_2} |z|_{H_0^1}^2 |v|^2
        + C_1 C_2 |v|_{H_0^1}^2
        + \frac{C_1}{C_2} |z|^2 |z|_{H_0^1}^2
        + C_1 C_2 |z|_{H_0^1}^2,
    \end{split}
\end{equation}
where $C_2$ is a positive constant which will be specified later.
Secondly we estimate nonlinear term $N$.
For all $r_1 >0$, we have
\begin{equation}
  \begin{split}
           -<N(z),v> \
     \leq & \mu_0 \epsilon^{- \alpha /2} |z|_{H_0^1} |v|_{H_0^1} \\
     \leq & r_1 |v|_{H_0^1}^2 + \frac{\mu_0^2 }{4r_1 \epsilon^{\alpha}} |z|_{H_0^1}^2.
  \end{split}
\end{equation}
Comprehensively,
\begin{equation}
  \begin{split}
      & \frac{1}{2} \frac{d}{dt}|v|^2
      +\frac{\lambda_1}{2}|v|^2
      +\frac{1}{2} |v|^2_V \\
    \leq & \frac{C_1 }{C_2} |z|_{H_0^1}^2 |v|^2
      + (C_1 C_2 +r_1)|v|_{H_0^1}^2
      + \frac{C_1}{C_2} |z|^2 |z|_{H_0^1}^2
      + C_1 C_2 |z|_{H_0^1}^2
      +  \frac{\mu_0^2 }{4r_1 \epsilon^{\alpha}} |z|_{H_0^1}^2,
  \end{split}
\end{equation}
where $\lambda_1$ is the first eigenvalue of operator $A$.
Let $g_1 = \frac{C_1}{C_2} |z|^2 |z|_{H_0^1}^2
      + C_1 C_2 |z|_{H_0^1}^2
      +  \frac{\mu_0^2 }{4r_1 \epsilon^{\alpha}} |z|_{H_0^1}^2$.
Then we have
\begin{equation}\label{Equation: estimate v diff form}
    \frac{d}{dt} |v(t)|^2
      + \left( \frac{1}{2}-\frac{C_1 C_2 +r_1 }{\lambda_1^{\frac{1}{2}}} \right) |v|_V^2
      + \left( \frac{\lambda_1}{2} - \frac{C_1 |z|_{H_0^1}^2}{C_2}  \right) |v|^2
    \leq g_1 .
\end{equation}
Choose  $C_2 < \frac{\lambda_1^{\frac{1}{2}}}{2 C_1}$
and $r_1$ small enough such that $ C_1 C_2 +r_1 < \frac{\lambda_1^{\frac{1}{2}}}{2}$.
We have
\begin{equation}\label{Equation: estimate v diff simp form }
    \frac{d}{dt} |v(t)|^2
      + \left( \frac{\lambda_1}{2}  - \frac{C_1 |z|_{H_0^1}^2}{C_2}  \right) |v|^2
    \leq g_1 .
\end{equation}
By Gronwall Lemma we have
\begin{equation}
  \begin{split}
    |v(t)|^2 \leq
     |v(0)|^2 e^{ -\int_0^t \left( \frac{\lambda_1}{2} - \frac{C_1 |z(s)|_{H_0^1}^2}{C_2} \right) ds }
     + \int_0^t g_1(s_1) e^{ -\int_{s_1}^t
       \left( \frac{\lambda_1}{2} - \frac{C_1 |z(s_2)|_{H_0^1}^2}{C_2} \right)ds_2 }ds_1.
  \end{split}
\end{equation}
Thus,
\begin{equation}
    \sup_{t \in [0,T]} |v(t)|^2
    \leq e^{\frac{C_1}{C_2} \int_0^T |z(s)|_{H_0^1}^2 ds} |u_0|^2
    + \int_0^T e^{\frac{C_1}{C_2} \int_s^T |z(r)|_{H_0^1}^2 dr} g_1(s) ds.
\end{equation}
Let $c_5= C_1/C_2$ and we obtain \eqref{Equation: estimate v in CH}.

Now integrating \eqref{Equation: estimate v diff form} over $[0,T]$, we get
\begin{equation}
    |v(T)|^2 - |v(0)|^2
     + \left( \frac{1}{2} - \frac{C_1 C_2 + r_1}{\lambda_1^{\frac{1}{2}}} \right) \int_0^T |v(s)|^2_V ds
    \leq \int_0^T \frac{C_1}{C_2} |z(s)|_{H_0^1}^2 |v(s)|^2 ds
     +\int_0^T g_1(s) ds.
\end{equation}
Let $c_6 = \left( \frac{1}{2} - (C_1 C_2 +r_1)\lambda_1^{-1/2} \right)^{-1} $
and inequality \eqref{Equation: estimate v in LV} follows.
The proof is complete.
\end{proof}
Since $z \in C([0,T];V)$, the following theorem is an immediate consequence of theorem \ref{Theorem: existence of solu}
and proposition \ref{Proposition: Extention of solu}.
\begin{theorem}
   For all $ T >0$ and $u_0 \in H$,
   the equation \eqref{Problem: differential form}
   has a unique solution $u \in C([0,T];H) \cap L^2(0,T;V)$ for all $\omega \in \Omega$
   in the sense of \eqref{Equation: integral form}.
\end{theorem}

\section{Random attractor}

In the sequel we aim to obtain a random attractor for RDS
generated by equation \eqref{Equation: integral form} under assumption (A4).
Denote
  \begin{align}
    Z(t) = Z(\theta_t \omega) = \int_{-\infty}^{t} S(t-r)dB^H (r)  \\
    Z(\omega)
     = \lim_{n \rightarrow \infty}  \int_{-n}^0 S(-r)dB^H (r).
  \end{align}
To show that the limit exists, we have
    \begin{equation}
      \begin{split}
       \mathbb{E} \left| Z(\omega)  \right|_{H_0^1}^2
      & = \mathbb{E} \left| \int_{-\infty}^0 S(-s)dB^H(s) \right|_{H_0^1}^2 \\
      & = \mathbb{E} \left| \lim_{t \rightarrow \infty} \sum_{i=1}^{\infty}
         \int_{-t}^0 S(-s) e_i d \beta_i^H (s) \right|_{H_0^1}^2 \\
      & \leq \limsup_{t \rightarrow \infty} \sum_{i}
         \mathbb{E} \left| \int_{-t}^0 K_H^* ( S(-\cdot) e_i)(s) d \beta_i (s) \right|_{H_0^1}^2 \\
      & = \limsup_{t \rightarrow \infty} \sum_{i}
        \left| K_H^* ( S(-\cdot) e_i) \right|_{L^2(-t,0;H_0^1)}^2  \\
      & = \limsup_{t \rightarrow \infty} \sum_{i}
        \left| S(\cdot) e_i \right|_{\mathcal{H}(0,t;H_0^1)}^2.
      \end{split}
    \end{equation}
    Since
    \begin{equation}
      \begin{split}
      \left| S(\cdot) e_i \right|_{\mathcal{H}(0,t;H_0^1)}^2
      & =
         \int_0^t \int_0^t <S(u)e_i,S(v)e_i >_{H_0^1} |u-v|^{2H-2} dudv  \\
      & =
         \int_0^t \int_0^t \lambda_i^{\frac{1}{2}} e^{-\lambda_i (u+v)} |u-v|^{2H-2} dudv  \\
      & = 2 \lambda_i^{\frac{1}{2}}
         \int_0^t \int_0^u  e^{-\lambda_i (2u-x)} x^{2H-2} dxdu  \\
      & = 2 \lambda_i^{\frac{1}{2}}
          \int_0^t \int_0^{\lambda_i u}
           e^{-\lambda_i 2u +y} \frac{y}{\lambda_i}^{2H-2} \frac{1}{\lambda_i} dydu  \\
      & = 2 \lambda_i^{\frac{3}{2}-2H}
          \int_0^{\lambda_i t} \int_{\frac{y}{\lambda_i}}^t
           e^{-\lambda_i 2u} e^y y^{2H-2}  dudy  \\
      & =  \lambda_i^{\frac{1}{2}-2H}
          \int_0^{\lambda_i t}
           e^y y^{2H-2} (e^{-2y}-e^{-2 \lambda_i t}) dy  \\
      & \leq  \lambda_i^{\frac{1}{2}-2H}
          \int_0^{\infty}
           e^{-y} y^{2H-2} dy
          = \lambda_i^{\frac{1}{2}-2H} \Gamma(2H-1),
      \end{split}
    \end{equation}
     we have
    \begin{equation}
     \begin{split}
       \mathbb{E} \left| Z(\omega)  \right|_{H_0^1}^2
       &\leq \Gamma(2H-1) \sum_{i=1}^{\infty} \lambda_i^{\frac{1}{2}-2H} \\
       &\leq 2 \Gamma(2H-1) \cdot \beta_D(4H-1) \cdot \xi(4H-1)
       < \infty
     \end{split}
    \end{equation}
as $H > \frac{1}{2}$.
Thus, $Z$ (the so-called fractional Ornstein-Uhlenback process) is the unique stationary solution
of the linear stochastic evolution equation
\begin{equation}
    dZ(t)=AZ(t)+dB^H(t), \quad t \in \mathbb{R}.
\end{equation}
We need the stationary process $Z$ to construct the RDS because
when we investigate the long-time behavior of solution we encounter a generalized integration
over time variable.
And ergodic theory can help us to convert it to a integration over sample space.
Namely, Consider the real-valued continuous function $|Z(\theta_{\cdot} \omega)|_{H_0^1}^2$,
we have $|Z(\theta_{\cdot} \omega)|_{H_0^1}^2 \in L^1(\Omega, P)$.
    Since $(\Omega, \mathcal{F}, \{ \theta(t) \}_{t\in \mathbb{R}} )$
    is the metric dynamical system, we can use the Birkhoff-Chintchin Ergodic Theorem to obtain
    \begin{equation}\label{Equation: Ergodic}
     \begin{split}
        \lim_{n\rightarrow \pm \infty} \frac{1}{n}
         \int_0^n | Z (\theta_t \omega) |_{H_0^1}^2 dt
        & = \mathbb{E} | Z (\omega) |_{H_0^1}^2 \\
        & \leq  2 \Gamma(2H-1) \cdot \beta_D(4H-1) \cdot \xi(4H-1).
     \end{split}
    \end{equation}
For all $t_0 \in \mathbb{R}$, by Theorem \ref{Theorem: existence of solu},
$u(t,\omega;t_0,u_0)$ is the unique solution of the equation
\begin{equation}\label{Equation: integral form u t t0}
    u(t;t_0)
    = S(t-t_0)u_0
     - \int_{t_0}^t S(t-s) B(u(s))ds
     - \int_{t_0}^t S(t-s) N(u(s))ds
     + \int_{t_0}^t S(t-s) dB^H(s).
\end{equation}
In this section, let $u(t,\omega;t_0)=v(t,\omega;t_0) + Z(t,\omega)$, we have
\begin{equation}
  \begin{split}
      & v(t)+ \int_{-\infty}^t S(t-s) dB^H(s) \\
    = & S(t)u_0
      -\int_{t_0}^t S(t-s) B( v(s) + Z(s) ) ds
      -\int_{t_0}^t S(t-s) N( v(s) + Z(s) ) ds
      + \int_{t_0}^t S(t-s) dB^H(s) .
  \end{split}
\end{equation}
Since
\begin{equation}
    \int_{-\infty}^{t_0} S(t-s)dB^H(s)
    = S(t-t_0) Z (\theta_{t_0}\omega),
\end{equation}
$v(t,\omega;t_0,u_0-Z(\theta_{t_0}\omega))$ is the unique solution of the integral equation
\begin{equation}\label{Equation: integral form for v infty}
    v(t)
    =S(t) (u_0-Z(\theta_{t_0}\omega))
     -\int_{t_0}^t S(t-s) B( v(s) + Z(s) ) ds
     -\int_{t_0}^t S(t-s) N( v(s) + Z(s) ) ds.
\end{equation}
Then $v$ is the weak solution of the following differential equation
\begin{align}
    & \frac{dv}{dt} + A(v+Z) + B(v+Z) = 0 \label{Equation: differ form for v infinity}\\
    & v(t_0) = u_0 - Z(\theta_{t_0}\omega) \label{Equation: boundary for v infinity}
\end{align}
We can now define an continuous mapping by setting
\begin{equation}
    \varphi(t,\omega,u_0)= v(t,\omega; 0,u_0-Z(\omega)) + Z(\theta_t \omega),
    \quad \forall (t,\omega,u_0) \in \mathbb{R} \times \Omega \times H.
\end{equation}
The measurability follows from the continuity dependence of solution with respect to initial value.
the cocycle property follows from the uniqueness of solution
for all noise path $\omega \in \Omega$.
Thus, $\varphi$ is a RDS associated with \eqref{Equation: integral form}.
In the rest of this section, we will compute some estimates in spaces $H$ and $V$.
Then we use these estimates and compactness of the embedding $V \hookrightarrow H$
to obtain the existence of a compact random attractor.

\begin{lemma}\label{Lemma: absorb in H}
  If $C_1^2 > \frac{ \lambda_1^{2/3} }{8 \Gamma(2H-1) \beta_D(4H-1)  \xi(4H-1)}$,
  then there exist random radii  $\rho_H(\omega)>0$ and $\rho_1(\omega)$ such that
  for all $ M>0$ there exists $t_2(\omega) < -1$,
  such that whenever $t_0 < t_2$ and $|u_0| <M$, we have
  \begin{align}
    | v(t, \omega; t_0,u_0 - Z(\theta_{t_0}\omega))|^2 &\leq \rho_H(\omega),
        \quad \forall \ t \in [-1,0] \\
    | u(t, \omega; t_0,u_0)|^2 &\leq \rho_H(\omega) ,
         \quad \forall  \ t \in [-1,0]\\
    \int_{-1}^0 |v(t)|_V^2 dt &\leq \rho_1(\omega)  \\
    \int_{-1}^0 |v(t)+Z(t)|_V^2 dt &\leq \rho_1(\omega)
  \end{align}
\end{lemma}

\begin{proof}
  The proof is base on Proposition \ref{Proposition: Extention of solu}.
  Let $g_2
      =\frac{C_1}{C_2} |Z|^2 |Z|_{H_0^1}^2
      + C_1 C_2 |Z|_{H_0^1}^2
      +  \frac{\mu_0^2 }{4r_1 \epsilon^{\alpha}} |Z|_{H_0^1}^2$.
  Similar to \eqref{Equation: estimate v diff form} we have
  \begin{equation}\label{Equation: estimate v infty diff form}
    \frac{d}{dt} |v|^2
      + \left( \frac{1}{2}-\frac{C_1 C_2 +r_1 }{\lambda_1^{\frac{1}{2}}} \right) |v|_V^2
      + \left( \frac{\lambda_1}{2} - \frac{C_1 |Z|_{H_0^1}^2}{C_2}  \right) |v|^2
    \leq g_2 .
  \end{equation}
By assumption we can choose
$C_2 \in (
  \frac{4 C_1  \Gamma(2H-1)  \beta_D(4H-1)  \xi(4H-1)}{\lambda_1},
   \frac{\lambda_1^{\frac{1}{2}}}{2C_1})$
and $r_1$ small enough and we have
\begin{equation}\label{Equation: estimate v infty diff simp form }
    \frac{d}{dt} |v|^2
      + \left( \frac{\lambda_1}{2}  - \frac{C_1 |Z|_{H_0^1}^2}{C_2}  \right) |v|^2
    \leq g_2 .
\end{equation}
By Gronwall inequality, when $t\in [-1,0]$ and $t_0<-1$, we have
\begin{equation}
  \begin{split}
      & |v(t)|^2  \\
     \leq & |v(t_0)|^2 e^{ -\int_{t_0}^{t} \left( \frac{\lambda_1}{2} - \frac{C_1 |Z(s)|_{H_0^1}^2}{C_2} \right) ds }
       + \int_{t_0}^{t} g_2(s_1) e^{ -\int_{s_1}^{t}
       \left( \frac{\lambda_1}{2} - \frac{C_1 |Z(s_2)|_{H_0^1}^2}{C_2} \right)ds_2 }ds_1 \\
     \leq & |v(t_0)|^2 e^{ -\int_{t_0}^{0} \left( \frac{\lambda_1}{2} - \frac{C_1 |Z(s)|_{H_0^1}^2}{C_2} \right) ds }
       + \int_{t_0}^{0} g_2(s_1) e^{ -\int_{s_1}^{0}
       \left( \frac{\lambda_1}{2} - \frac{C_1 |Z(s_2)|_{H_0^1}^2}{C_2} \right)ds_2 }ds_1.
  \end{split}
\end{equation}
Due to the ergodic property of fractional O-U process \eqref{Equation: Ergodic}, we have
\begin{equation}
    \lim_{t_0 \rightarrow -\infty} \frac{1}{-t_0}
    \int_{t_0}^0 |Z(s)|^2_{H_0^1} ds
    = \mathbb{E} |Z(\omega)|_{H_0^1} .
\end{equation}
Choose $r_2$ small enough such that
\begin{equation}
    \frac{C_1}{C_2} \mathbb{E} |Z(\omega)|_{H_0^1}
    \leq \frac{C_1}{C_2}
    2 \Gamma(2H-1) \beta_D(4H-1) \xi (4H-1)
    < \frac{\lambda_1}{2} -r_2.
\end{equation}
Then there exists $t_1(\omega) < -1$, such that when $t_0 < t_1$ we have
\begin{equation}
    |v(t)|^2
    \leq e^{(1+t_0)r_2} |u_0|^2
    + \int_{t_0}^{0} e^{(1+t_0)r_2} g_2(s) ds,
         \quad \forall t \in [-1,0].
\end{equation}
By Lemma 2.6 of \cite{Maslowski Schmalfus},
$g_2$ has at most polynomial growth as $t_0 \rightarrow -\infty$ for P-a.s.
$\omega \in \Omega$.
Thus, we have
\begin{equation}
    \int_{t_0}^{0} g_2(s) e^{(1+s)r_2} ds
    \leq \int_{-\infty}^{0} g_2(s) e^{(1+s)r_2} ds
    \leq \infty, \quad \text{P-a.s.}
\end{equation}
Let $\rho_H = 4\int_{-\infty}^{0} g_2(s) e^{(1+s)r_2} ds + 2 \sup_{t \in [-1,0]} |Z(t)|^2$
and there exists $t_2(\omega) < t_1(\omega) < -1$ such that for all $|u_0| \leq M$
\begin{align}
    |v(-1,\omega ; t_0,u_0-Z(\theta_{t_0}\omega))|^2
      &\leq 2\int_{-\infty}^{0} g_2(s) e^{(1+s)r_2} ds \\
    \begin{split}
    |u(-1,\omega ; t_0,u_0)|^2
      &\leq 2 |v(-1,\omega ; t_0,u_0-Z(\theta_{t_0}\omega))|^2
      + 2 \sup_{t \in [-1,0]} |Z(t)|^2 \\
      &\leq \rho_H(\omega), \quad  \forall t_0 < t_2, t \in [-1,0].
    \end{split}
\end{align}
In the following we consider a bound of $\int_{-1}^0 |v(t)|_V^2 dt$.
Integrating \eqref{Equation: estimate v infty diff form} over $[-1,0]$ we have
\begin{equation}
    |v(0)|^2 - |v(-1)|^2 + c_6^{-1} \int_{-1}^0 |v(t)|_V^2 dt
    \leq \int_{-1}^{0} g_2(t)dt + \int_{-1}^0 ( \frac{C_1}{C_2} |Z(t)|_1^2 )|v(t)|^2 dt.
\end{equation}
When $t_0<t_2$ we have
\begin{equation}
      \int_{-1}^0 |v(t)|_V^2 dt
     \leq c_6 (
      \int_{-1}^{0} g_2(t)dt
      +  \frac{C_1 \rho_H}{C_2} \int_{-1}^0 |Z(t)|_1^2  dt
      + |v(-1)|^2)
     \triangleq C(\omega).
\end{equation}
Similarly,
\begin{equation}
      \int_{-1}^0 |v(t)+Z(t)|_V^2 dt
     \leq 2c_6 (
      \int_{-1}^{0} g_2(t)dt
      +  \frac{C_1 \rho_H}{C_2} \int_{-1}^0 |Z(t)|_1^2  dt
      + 2 \int_{-1}^0 |Z(t)|_V^2 dt)
     \triangleq \widetilde{C}(\omega) .
\end{equation}
Let $\rho_1(\omega) = \max \{C(\omega),\widetilde{C}(\omega)\}$ and the proof is complete.
\end{proof}

\begin{lemma}\label{Lemma: absorb in V}
  Under the assumption of Lemma \ref{Lemma: absorb in H},
  there exists a random radius $\rho_V(\omega)$ such that for all $M>0$
  and $|u_0| < M $, there exists $t_2(\omega)<-1$ such that P-a.s.
  \begin{align}
     |v(t,\omega;t_0,u_0 - Z(\theta_{t_0}\omega))|_1^2
        & \leq \rho_V(\omega), \\
     |u(t,\omega;t_0,u_0 )|_1^2
        & \leq \rho_V(\omega), \quad \forall t_0 < t_2, t \in [-\frac{1}{2} ,0].
  \end{align}
\end{lemma}

\begin{proof}
  Integrating \eqref{Equation: differ form for v infinity} with $- \triangle v$ over $\mathcal{O}$,
  we have
  \begin{equation}\label{Estimate: in V}
     \frac{1}{2} \frac{d}{dt} |v|_1^2 + |v|_3^2
    \leq |b(v+Z,v+Z, \triangle v)| - <N(v+Z), -\triangle v>
  \end{equation}
  By using the Gagliardo-Nirenberg inequality and Young inequality,
  we have
  \begin{equation}
    \begin{split}
           & |b(v+Z,v+Z,\triangle v)| \\
      \leq & C |v+Z|^{1/2} \cdot |v+Z|_2^{1/2} \cdot  |v+Z|_1  \cdot |v|_2 \\
      \leq & C |v+Z|^{1/2} \cdot  |v+Z|_1  \cdot |v|_2^{3/2}
        + C |v+Z|^{1/2}  \cdot |v+Z|_1  \cdot |Z|_2^{1/2} \cdot  |v|_2 \\
      \leq & \frac{\lambda_1}{8} |v|_1^2
        + \frac{54C^4}{\lambda_1^3} |v+Z|^2  \cdot |v+Z|_1^4
        + \frac{\lambda_1}{8} |v|_2^2
        + \frac{2C^2}{\lambda_1} |v+Z| \cdot |v+Z|_1^2  \cdot |Z|_2
    \end{split}
  \end{equation}
  Inspired by \cite{GuoGuo}, we have the following estimate
  \begin{equation}
    \begin{split}
         -<N(v+Z, \triangle v)>
      \leq & \mu_0 \epsilon^{-\frac{\alpha}{2}} \int_{\mathcal{O}} |e_{ij}(v+Z)e_{ij}(\triangle v)| dx \\
      \leq & \mu_0 \epsilon^{-\frac{\alpha}{2}} |v+Z|_1 \cdot  |v|_3 \\
      \leq & \frac{1}{4} |v|_3^2 + \frac{\mu_0^2}{\epsilon^{\alpha}} |v+Z|_1^2 .
    \end{split}
  \end{equation}
  Let
  \begin{align}
    g_3(t)
     &= \frac{54 C }{\lambda_1^3} |v+Z|^2  \cdot |Z|_1^4
     + \frac{ 2C^2 }{\lambda_1} |v+Z| \cdot |v+Z|_1^2  \cdot |Z|_2
     + \frac{\mu_0^2}{\epsilon^{\alpha}} |v+Z|_1^2 ,\\
    g_4(t)
     &= \frac{54 C }{\lambda_1^3} |v+Z|^2  \cdot |v|_1^2.
  \end{align}
  Rewrite \eqref{Estimate: in V} and we have
  \begin{equation}
    \frac{d}{dt} |v(t)|_1^2
    \leq g_3(t) + g_4(t) |v(t)|_1^2
  \end{equation}
  By the variation of constant formula,
  for all $-1 \leq s \leq t \leq 0$ we have
  \begin{equation}
   \begin{split}
    |v(t)|_1^2
      & \leq |v(s)|_1^2 \cdot e^{\int_s^t g_4(s_1)ds_1} +
        e^{\int_s^t g_4(s_1)ds_1} \cdot \int_s^t g_3(s_2) e^{-\int_{s}^{s_2} g_4(s_1)ds_1} ds_2 \\
      & \leq \left( |v(s)|_1^2  + \int_{-1}^0 g_3(s_2)ds_2 \right)
        \cdot e^{\int_{-1}^0 g_4(s_1)ds_1}
   \end{split}
  \end{equation}
  Integrating $s$ over $[-1,t]$ we have
  \begin{equation}
    (1+t)|v(t)|_1^2
    \leq \left( \int_{-1}^{0}|v(s)|_1^2 ds + \int_{-1}^0 g_3(s)ds \right)
      \cdot e^{\int_{-1}^0 g_4 (s)ds}
  \end{equation}
  Notice that all the terms of $\int_{-1}^{0} g_3(s)ds $,
  $\int_{-1}^{0} g_4(s)ds $ and $\int_{-1}^{0}|v(s)|_1^2 ds$ are bounded
  as $t_0 \rightarrow -\infty$.
  Therefore,
  \begin{equation}
    |v(t)|_1^2 \leq C(\omega), \quad \forall t_0 < t_2, t \in [-\frac{1}{2},0].
  \end{equation}
  And there exists a random radius $\rho_V(\omega)$ such that
  \begin{equation}
    \begin{split}
       |u(t,\omega;t_0,u_0 )|_1^2
        & \leq 2 |v(t,\omega;t_0,u_0 - Z(\theta_{t_0}\omega))|_1^2
          + \sup_{t \in [-\frac{1}{2},0]} |Z(t)|_1^2 \\
        & \leq \rho_V(\omega), \quad \forall t_0 < t_2, t \in [-\frac{1}{2} ,0].
    \end{split}
  \end{equation}
  Especially when $t=0$, we have
  \begin{equation}
       |u(0,\omega;t_0,u_0 )|_1^2
        \leq \rho_V(\omega), \quad \forall t_0 < t_2.
  \end{equation}
\end{proof}

Denote $\dot{H}^1 = $ the closure of $\mathcal{V}$ in $\left( H^1(\mathcal{O})\right)^2$.
Lemma \ref{Lemma: absorb in V} shows that
there exists a bounded random ball in $\dot{H}^1$
which absorbs any bounded non-random subset of $H$.
Since $\dot{H}^1$ is compactly embedded in $H$,
we have establish the existence of a compact random absorbing set in $H$.
By Proposition \ref{Proposition: exi rand attr} we state our the final theorem:

\begin{theorem}
  If
$C_1^2 > \frac{ \lambda_1^{2/3} }{8 \Gamma(2H-1) \beta_D(4H-1)  \xi(4H-1)}$,
  the random dynamical system associated with \eqref{Equation: integral form}
  has a random attractor.
\end{theorem}

\section*{Acknowledgments}
The authors are grateful to Mar\'ia J. Garrido-Atienza, Bohdan Maslowski and Bj\"orn Schmalfuss for their kindly discussions and suggestion 
concerning the general framework of stochastic equations driven by fractional Brownian motion.

%acknowledge support by National Natural Science
%Foundation of China (Grant No. 10571175) and SRF for ROCS, SEM, China.

% \bibliographystyle{amsplain} \bibliography{per.homoc.hetero}

\vspace{1cm}
{\large Jin Li}  (the corresponding author)

Email: talent\_tim@yahoo.com.cn \\

{\large Jianhua Huang}

Email: jhhuang@nudt.edu.cn

\end{document}